\documentclass[10pt]{article}
\usepackage[english]{babel}										
\usepackage[utf8]{inputenc}										
\usepackage[T1]{fontenc}										
\usepackage{makecell}
\usepackage{amsmath,amsfonts,amssymb,amsthm,cancel,siunitx,
calculator,calc,mathtools,empheq,latexsym}

\usepackage{subfig,epsfig,tikz,float}		            
\usepackage{booktabs,multicol,multirow,tabularx,array}          
\usepackage[utf8]{inputenc} 
\usepackage[T1]{fontenc}    
\usepackage{hyperref} 
\usepackage{subfig}
\usepackage{url}            
\usepackage{booktabs}       
\usepackage{amsfonts}       
\usepackage{nicefrac}       
\usepackage{microtype}      
\usepackage{bm}             
\usepackage{graphicx}       
\usepackage{subfig}
\usepackage{algorithm}      
\usepackage[noend]{algpseudocode} 
\usepackage{caption}
\usepackage[font=small,labelfont=bf,skip=0.5em]{caption}
\usepackage{lmodern}
\usepackage{array}          
\usepackage{booktabs}       
\usepackage{pbox}           
\usepackage{amsmath} 
\usepackage{listings}
\usepackage{tikz}

\usepackage{mathtools}
\usepackage{subfig} 
\usepackage{todonotes}
\usepackage{amsthm,comment}
\usetikzlibrary{shapes.geometric, arrows}

\newtheorem{theorem}{Theorem}
\newtheorem{prop}{Proposition}
\newtheorem{lemma}{Lemma}
\newtheorem{definition}[theorem]{Definition}

\title{\normalsize\bf%
\uppercase{Dynamical analysis of an HIV infection model including quiescent cells and immune response}
}
\author{%
Ibrahim Nali$^{1*}$,  Attila Dénes$^{1,2} $, Abdessamad Tridane$^{3}$ and Xueyong Zhou$^{4}$ 
}

\begin{document}

\date{}

\maketitle

\vspace{-0.5cm}

\begin{center}
{\footnotesize 
*Corresponding author\\
$^1$Bolyai Institute, University of Szeged, Szeged 6720, Hungary \\
$^2$National Laboratory for Health Security, Hungary \\
$^3$Emirates Center for Mobility Research, United Arab Emirates University, Al Ain 15551, United Arab Emirates \\
$^4$School of Mathematics and Statistics, Xinyang Normal University,Xinyang 464000, Henan, China\\
E-mails: ibrahim.nali@usmba.ac.ma
}
\end{center}
\bigskip
\noindent
\abstract{
This research gives a thorough examination of an HIV infection model that includes quiescent cells and immune response dynamics in the host. The model, represented by a system of ordinary differential equations, captures the complex interaction between the host's immune response and viral infection. The study focuses on the model's fundamental aspects, such as equilibrium analysis, computing the basic reproduction number $\mathcal{R}_0$, stability analysis, bifurcation phenomena, numerical simulations, and sensitivity analysis.

The analysis reveals both an infection equilibrium, which indicates the persistence of the illness, and an infection-free equilibrium, which represents disease control possibilities. Applying matrix-theoretical approaches, stability analysis proved that the infection-free equilibrium is both locally and globally stable for $\mathcal{R}_0 < 1$. For the situation of $\mathcal{R}_0 > 1$, the infection equilibrium is locally asymptotically stable via the Routh--Hurwitz criterion. We also studied the uniform persistence of the infection, demonstrating that the infection remains present above a positive threshold under certain conditions. The study also found a transcritical forward-type bifurcation at $\mathcal{R}_0 = 1$, indicating a critical threshold that affects the system's behavior. The model's temporal dynamics are studied using numerical simulations, and sensitivity analysis identifies the most significant variables by assessing the effects of parameter changes on system behavior.
}
\medskip
\noindent
{\small{\bf Keywords}{: Virus dynamics; Quiescent cells; Immune response; Stability analysis; Lyapunov function, Bifurcation.} 

}

\baselineskip=\normalbaselineskip

\section{Introduction}\label{sec:1}
Human immunodeficiency virus (HIV) affects millions of individuals worldwide. The body's fight against infections depends heavily on CD4+ T cells, which are the precise target of this substance. Understanding the intricate dynamics of HIV infection has been made possible by the use of within-host models \cite{BOCHAROV, Dorratoltaj, Perelson}. These models offer a framework for researching the complex relationships between various immune system components, viral replication, and the disease's course in an infected person. Quiescent cells, or resting cells, are an essential component of within-host models and play a major role in the latency and persistence of HIV infection \cite{Alizon, DeBoer}. These cells are essential in determining the long-term course of HIV infection because they act as reservoirs for the virus even in the midst of antiretroviral therapy. Researchers can learn more about the dynamics of quiescent cells and how they interact with the immune system by including them in the models \cite{pearce}. Additionally, the immunological response, which involves intricate interactions between immune cells, antibodies, and cytokines, plays a vital role in combating HIV. The way the infection progresses depends on the dynamic interaction of the virus, immune system, and quiescent cells, which affects viral load, reservoir formation, and disease development rates \cite{fauci}. 

Moreover, although they have ineffective infection rates, dormant CD4+ T cells have been identified as possible contributors to the viral reservoir \cite{vatakis}. These cells are essential to the dynamics of HIV, and the way they interact with the immune system creates a complicated picture. Prior research indicates that immunological quiescence, defined as low immune activation \cite{card}, may function as a barrier to HIV infection by reducing the number of target cells available. Nevertheless, substantial numbers of improperly processed viral ends and abortive circles are involved in the integration of HIV into quiescent CD4+ T cells \cite{vatakis1}, impeding the infection process. Potential targets for treatment development are provided by the identification of particular biological proteins in quiescent cells that prevent HIV infection\cite{zack}.

The influence of viral variety and its consequences on treatment results and disease progression is another important aspect of HIV infection \cite{Lattaide}. HIV's rapid replication rate and error-prone replication process result in considerable genomic diversity. Due to this diversity, an infected human may contain several virus strains, often known as quasispecies \cite{Liu, Shan}. The virus can adapt and elude immune responses due to the constant production of new versions, which makes the development of antiretroviral treatments and vaccines difficult. Mathematical models that incorporate viral variety provide useful insights into the evolutionary dynamics of HIV and its repercussions on the host-virus interaction, various models have been introduced \cite{marcuk,Iwami,warmik, shwartz,Veronica}. With the aid of these models, researchers can investigate methods to reduce the viral variety and its possible effects on treatment efficacy, as well as the formation and spread of drug-resistant strains and the influence of immune selection pressure. To improve therapeutic approaches and patient outcomes in HIV infection, it is crucial to comprehend the intricate interactions among immune responses, treatment results, and viral variety.

We present and assess a model that builds on Kouche et al.~\cite{Kouche} and Pang et al. \cite{pang}'s advances in the field. Kouche et al.~\cite{Kouche} updated Guedj et al.~\cite{Guedj}'s model, taking into account RT inhibitors and quiescent cells. Their findings showed that increasing the drug's potency or extending its active period could help eliminate the infection faster. Pang et al.~\cite{pang} updated Nowak et al.'s model \cite{Nowak} by utilizing a Holling type 2 function to classify immunological responses. Their contributions  served as a motivation for our research, encouraging a thorough dynamical analysis of an HIV infection model that includes quiescent cells and the immune response. We hope that our research will provide useful insight into the complex interaction between the virus and the host's defense mechanisms. Given the scarcity of works on the mathematical analysis of HIV dynamic systems that include the quiescent cell state.

Our paper is organized as follows: in the next section, we present our proposed model, which incorporates the dynamics of HIV infection, quiescent cells, and the immune response. We thoroughly examine the mathematical well-posedness of the model in Section 3. Section 4 is dedicated to the study of equilibria, where we analyze the Infection-free equilibrium and calculate the reproduction number. Subsequently, in the following section, we delve into the stability analysis of both the Infection-free and infection equilibria, shedding light on the long-term behavior of the system. Furthermore, we investigate the existence of transcritical bifurcation, in the subsequent section. In Section 6, we present our numerical analysis to validate our theoretical findings. Finally, we conclude the paper with a concise discussion, summarizing our key findings, highlighting the implications of our research, and identifying potential avenues for future exploration.

\section{Model derivation}
The model is determined by the following system of equations:
\begin{equation}
\label{1}
\begin{aligned}
&\frac{d Q}{d t}=\Lambda+\varrho T-\sigma Q-\mu_1 Q, \\
&\frac{d T}{d t}=\sigma Q-\beta T V-\varrho T-\mu_2 T, \\
&\frac{d I}{d t}=\beta T V-\mu_3 I-p I Z, \\
&\frac{d V}{d t}=\xi I-\mu_4 V, \\
&\frac{d Z}{d t}=b+\frac{c I Z}{\kappa+I}-\mu_5 Z.
\end{aligned}
\end{equation}
This is a within-host model represented by a system of ordinary differential equations, where the state variables represent different components of the host immune response to a viral infection. The variables and their corresponding meanings are as follows:
\begin{itemize}
    \item[$Q$:] quiescent cells, representing a dormant or inactive population of cells;
    \item[$T$:] healthy activated cells, representing cells that have become activated in response to the viral infection;
    \item [$I$:]  infected cells, representing cells that have been infected by the virus;
    \item [$V$:] free virus, representing the number of virus particles circulating in the host;
    \item [$Z$:] CTL (cytotoxic T-lymphocyte) cells, representing the number of immune cells that specifically target and kill infected cells.
\end{itemize}
We denote by $\Lambda$ the rate of influx of quiescent cells, $\varrho$ stands for the rate of transition of healthy activated cells back to quiescent cells, and $\sigma$ for the rate of transition of quiescent cells to healthy activated cells. We introduce the notations $\mu_1$, $\mu_2$, and $\mu_3$ for the natural death rates of quiescent cells, healthy activated cells, and infected cells, respectively. The notation $\beta$ represents the rate of infection of healthy activated cells by free virus. The rate of killing of infected cells by CTL cells is denoted by $p$, while $\xi$ denotes the rate of production of free virus from infected cells. We introduce $\mu_4$ for the rate of clearance of free virus and $\mu_5$ for the natural death rate of CTL cells. For the immune response equation, $b$ is the parameter representing the rate of production of CTL cells, and $c$ denotes the rate at which the CTLs are stimulated and activated in response to the presence of infected cells. To model the activation and proliferation of CTL cells, we use the function $\frac{c I Z}{\kappa + I}$, where $\kappa$ is the virus load required for half-maximal CTL cell stimulation \cite{Gennady}. This function reflects both antigenic stimulation and the export of specific precursor CTL cells from the thymus, providing a more accurate representation of immune dynamics. Unlike the traditional bilinear term used in earlier models \cite{pang}, this saturating function better captures the immune system's behavior by considering non-cytolytic mechanisms, where CTL cells control the infection without directly killing infected cells. The parameter $\kappa$ plays a crucial role in balancing the immune response, reflecting how CTL stimulation adjusts according to the viral load \cite{Gennady}.

\begin{center}
\begin{table}[h!]%
\label{tab1}
\centering
\caption{Parameters of  model \eqref{1}.\label{tab1}}%
\begin{tabular*}{370pt}{@{\extracolsep\fill}cllll@{\extracolsep\fill}}
\toprule
\textbf{Parameter} & \textbf{Definition}   \\
\midrule
$\Lambda$ & {Rate of influx of quiescent cells}     \\
$\varrho$  & \text{Rate of transition of healthy activated cells back to quiescent cells}     \\
$\sigma$  & \text{Rate of transition of quiescent cells to healthy activated cells}    \\
$\mu_1$ & \text{Natural death rate of quiescent cells} \\
$\beta$ & \text{Rate of infection of healthy activated cells by free virus} \\
$\mu_2$ & \text{Natural death rate of healthy activated cells} \\
$\mu_3$  &  \text{Rate of death of infected cells} \\
$p$ & \text{Rate of killing of infected cells by CTL cells } \\
$\xi$  &  \text{Rate of production of free virus from infected cells} \\ 
$\mu_4$  & \text{Rate of clearance of free virus} \\
$b$ &  \text{Rate of production of CTL cells} \\
$c$ & \makecell[l]{Rate at which the CTLs are stimulated and activated\\ in response to the presence of infected cells}\\
$\kappa$ & \text{Represents the virus load for half-maximal CTL cell stimulation.}\\
$\mu_5$ & \text{Natural death rate of CTL cells} \\
\bottomrule
\end{tabular*}
\end{table}
\end{center}
The parameters of the model and their definitions are summarized in Table \ref{tab1}, while
the flow chart depicting the within-host dynamics of the model is shown in Figure  \ref{fig:flowchart}.
\begin{figure}[h!]
 \centerline{\includegraphics[scale=0.5]{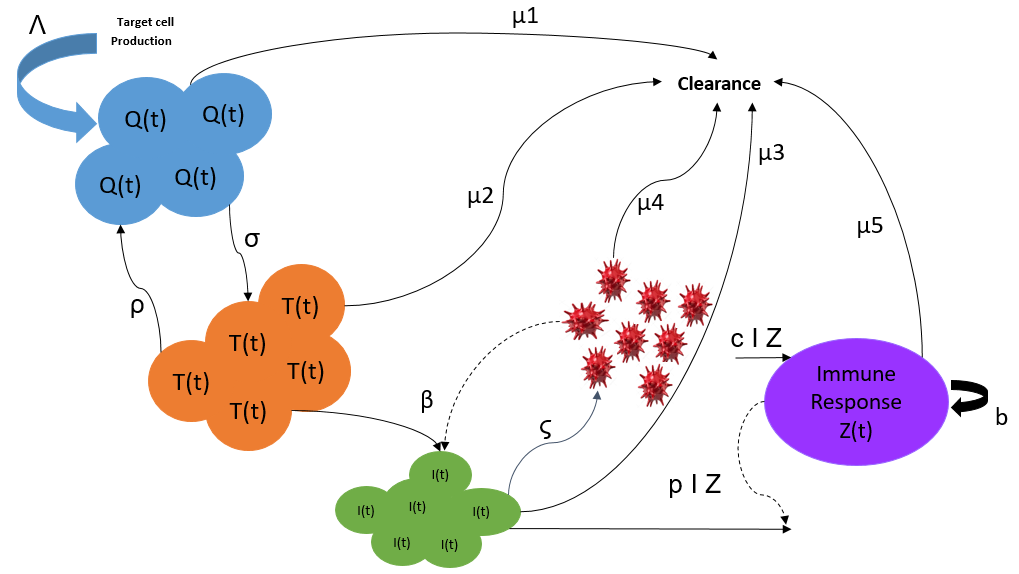}}
    \caption{Flow chart of the model \eqref{1}}
    \label{fig:flowchart}
\end{figure}

\section{Positivity and boundedness of the solutions}
We  begin our analysis of system \eqref{1} by examining some fundamental properties of the model. Firstly, we will show that any solution of \eqref{1}, initiated from a non-negative initial condition in $\mathbb{R}_{+}^{5}$, will remain non-negative. Specifically, we can see that
$$
\begin{aligned}
\left.\frac{d Q}{d t}\right|_{Q=0} &=\Lambda+\varrho T >0 \quad \text{for all } T \geq 0,\\
\left.\frac{d T}{d t}\right|_{T=0} &= \sigma Q \geq 0 \quad \text{for all } Q \geq 0, \\
\left.\frac{d I}{d t}\right|_{I=0} &= \beta T V \geq 0 \quad \text{for all } T, V \geq 0, \\ \left.\quad \frac{d V}{d t}\right|_{V=0}&= \xi I \geq 0 \quad \text{for all } I \geq 0, \\
\left.\frac{d Z}{d t}\right|_{Z=0} &= b > 0.
\end{aligned}
$$
This proves that $\mathbb{R}_{+}^{5}$ is positively invariant with respect to system \eqref{1}, meaning that any solution of \eqref{1} will remain in $\mathbb{R}_{+}^{5}$ for all times.

We aim to show that the solutions of the system are bounded, under the assumption  $\mu_5 \geq c$. Let $W(t)=Q(t)+T(t)+ I(t)+\frac{\mu_3}{2 \xi} V(t)+Z(t)$, the derivative of which can be computed as
$$
\begin{aligned}
\dot{W} & =\dot{Q}+\dot{T}+\dot{I}+\frac{\mu_3}{2 \xi} \dot{V}+\dot{Z} \\
& =\Lambda+b-\mu_1 Q-\mu_2 T-\frac{\mu_3}{2} I-p I Z-\frac{\mu_3 \mu_4}{2 \xi} V+\frac{c I Z}{\kappa+I}-\mu_5 Z \\
& \leq \Lambda+b-\mu_1 Q-\mu_2 T-\frac{\mu_3}{2} I-\frac{\mu_3 \mu_4}{2 \xi} V+\left(c-\mu_5\right) Z \\
& \leq \Lambda+b-m\left(Q+T+I+\frac{\mu_3}{2 \xi} V+Z\right)=\xi-m W,
\end{aligned}
$$
where we  define $m=\min  \left\{\mu_1,\mu_2,\frac{\mu_3}{2},\mu_4,\mu_5-c \right\}$ and $\xi=\Lambda+b$. In particular, the inequality uses the fact that $ p, \kappa, Z(t)$ and $I(t) $ are all nonnegative.

We can rewrite the inequality as $\dot{W}+m W \leq \xi$, which is a linear ordinary differential inequality with a negative coefficient for $W$. Hence, by applying the integrating factor $e^{m t}$, we get:
$$
W(t) \leq e^{-m t}\left(W(0)-\frac{\xi}{m}\right)+\frac{\xi}{m}.
$$
 This implies that $0 \leq W(t) \leq M_1$, where $M_1=\frac{\xi}{m}$, and thus $0 \leq Q(t), T(t), I(t),Z(t) \leq M_1$ and $0 \leq V(t) \leq M_2$ for all $t \geq 0$, provided that the initial conditions satisfy $Q(0)+T(0)+ I(0)+\frac{\mu_3}{2 \xi} V(0)+Z(0) \leq M_1$. Here, we have defined $M_2=\frac{2 \xi M_1}{\mu_3}$.Thus, we have demonstrated that the solutions are bounded. Consequently, we can assert that the set
$$
\mathcal{D} := \left\{ (Q,T, I, V,Z) \in \mathbb{R}^5_+ \mid Q+T + I+Z \leq 3 M_1, \; V \leq M_2 \right\}
$$
is  positively invariant.
\section{Equilibria, reproduction number}
\subsection{Infection-free equilibrium point}
To find the equilibria of \eqref{1}, one has to solve the algebraic system of equations
\begin{equation}\label{algeq}
\begin{split}
        0&= \Lambda+\varrho T^*-\sigma Q^*-\mu_1 Q^*, \\
 0&=\sigma Q^*-\beta T^* V^*-\varrho T^*-\mu_2 T^*, \\
0& =\beta T^* V^*-\mu_3 I^*-p I^* Z^*, \\
0& =\xi I^*-\mu_4 V^*, \\
0&=b+\frac{c I^* Z^*}{\kappa+I^*}-\mu_5 Z^*.
\end{split}
\end{equation}
The unique infection-free equilibrium of \eqref{1} can  be calculated as
$$\mathcal{E}_0=\left(\frac{\Lambda  (\mu_2+\varrho )}{\mu_1 (\mu_2+\varrho )+\mu_2 \sigma },\frac{\Lambda  \sigma }{\mu_1 (\mu_2+\varrho )+\mu_2 \sigma },0,0,\frac{b}{\mu_5}\right).$$

 \subsection{Basic reproduction number \texorpdfstring{($\mathcal{R}_0$)}{R0}}
To determine the basic reproduction number ($\mathcal{R}_0$), we use the next-generation matrix method (see e.g.\ \cite{diekmann}). The transmission matrix $f$  and  the transition matrix $v$  are given by
$$
f=\left[\begin{array}{c}
\beta T V \\
0
\end{array}\right] \text {\quad and\quad} v=\left[\begin{array}{c}
-\mu_3 I-p I Z \\
\xi I-\mu_4 V
\end{array}\right].
$$
The basic reproduction number is obtained from the spectral radius (dominant eigenvalue) of $F V_1^{-1}$, where $F$ and $V_1$ are the Jacobian matrices of $f$ and $v$ evaluated at the infection-free equilibrium point.

The matrices $F$ and $V_1$ for the model are given by
\begin{align*}
F&=\left(
\begin{array}{cc}
 0 & \frac{\beta  \Lambda  \sigma }{\mu_1 (\mu_2+\varrho )+\mu_2 \sigma } \\
 0 & 0 \\
\end{array}
\right)
\shortintertext{and}
V_1&=
\begin{pmatrix}
 -\frac{b p+\mu_3 \mu_5}{\mu_5} & 0 \\
 \xi & -\mu_4 \\
\end{pmatrix},
    \end{align*}
hence,
$$
FV^{-1}=\left(
\begin{array}{cc}
 -\frac{\beta  \Lambda  \mu_5 \sigma  \xi}{\mu_4 (b p+\mu_3 \mu_5) (\mu_1 (\mu_2+\varrho )+\mu_2 \sigma )} & -\frac{\beta  \Lambda  \sigma }{\mu_1 \mu_4 (\mu_2+\varrho )+\mu_2 \mu_4 \sigma } \\
 0 & 0 \\
\end{array}
\right),
$$
and the spectral radius of this matrix can be calculated as
$$
\mathcal{R}_{0}= \frac{\beta  \Lambda  \mu_5  \sigma  \xi}{\mu_4 (b p+\mu_3 \mu_5) (\mu_1 (\mu_2+\varrho )+\mu_2 \sigma )}.
$$
\subsection{Infection equilibrium}
By solving the system of algebraic equations to find infection equilibria of \eqref{algeq}, one obtains
\begin{equation}\label{equilibrium}
\begin{split}
V^*&=\frac{\xi}{\mu_4} I^*, \\
Z^*&=\frac{b(\kappa+I^*)}{\mu_5 \kappa + I^* (\mu_5-c)},\\
T^*&=\frac{\mu_3 \mu_4}{\beta \xi}+ \frac{b p \mu_4 (\kappa+I^*)}{(\mu_5 \kappa + I^* (\mu_5-c)) \beta \xi}, \\
Q^*&=\frac{1}{\mu_1 +\sigma } \left(\varrho \left[ \frac{\mu_3 \mu_4}{\beta \xi} + \frac{b p \mu_4 (\kappa+I^*)}{(\mu_5 \kappa + I^* (\mu_5-c)) \beta \xi} \right] +  \Lambda \right).
\end{split}
\end{equation}
If we replace these expressions into the second equation of \eqref{algeq}, we get
\begin{equation}
A I^{* 2}+B I^*+C=0,
\end{equation}
where
$$
\begin{aligned}
A={} & -\beta  \xi (\mu_1+\sigma ) (b p+\mu_3 (\mu_5-c)), \\
B= {}& -\beta  \xi (b \kappa  p (\mu_1+\sigma )+\Lambda  \sigma  (c-\mu_5)+\kappa  \mu_3 \mu_5 (\mu_1+\sigma ))\\&+\mu_4 (\mu_1 (\mu_2+\varrho )+\mu_2 \sigma ) (\mu_3 (c-\mu_5)-bp), \\
={}&\mu_4 (\mathcal{R}_{0}-1) (b p+\mu_3 \mu_5) (\mu_1 (\mu_2+\varrho )+\mu_2 \sigma )-\beta  \xi (\kappa  (\mu_1+\sigma ) (b p+\mu_3 \mu_5)+c \Lambda  \sigma )\\&+c \mu_3 \mu_4 (\mu_1 (\mu_2+\varrho )+\mu_2 \sigma ),\\
C={}& \beta  \kappa  \Lambda  \mu_5 \sigma  \xi-\kappa  \mu_4 (b p+\mu_3 \mu_5) (\mu_1 (\mu_2+\varrho )+\mu_2 \sigma )\\= {}&\kappa \mu_4 (b p+\mu_3 \mu_5) (\mu_1 (\mu_2+\varrho )+\mu_2 \sigma )(\mathcal{R}_{0} - 1 ).
\end{aligned}
$$ 
We assume that $\mu_5 > c$ and $\mathcal{R}_{0} > 1$, which implies that $B^2-4 A C > 0$. Under these conditions, the equation has two real roots given by
\begin{align*}
 I^*_{-}& = \frac{- B + \sqrt{B^2-4 A C}}{2 A}, \\
 I^*_{+}& = \frac{- B - \sqrt{B^2-4 A C}}{2 A}. 
 \end{align*}

Additionally, since $I^*_{+} \times I^*_{-} = \frac{\kappa \mu_4 (1-\mathcal{R}_{0}) (b p+\mu_3 \mu_5) (\mu_1 (\mu_2+\varrho )+\mu_2 \sigma )}{\beta \xi (\mu_1+\sigma ) (b p+\mu_3 (\mu_5-c))} < 0$, it follows that $I^*_{+}$ and $I^*_{-}$ have opposite signs, and It is clear that $I^*_{+} > 0$. Hence the system \eqref{1} has a unique positive equilibrium 
$$\mathcal{E}_1= (Q^*,T^*,I^*,V^*,Z^*).$$
Alternatively, if $\mu_5 < c < \frac{b p}{\mu_3} + \mu_5$ and $\mathcal{R}_{0} > 1$ , additional conditions are required to ensure equilibrium positivity. Under this condition, we enforce $I^*_+ < I_m =\frac{\mu_5 \kappa}{c-\mu_5}$. This guarantees the existence of a unique positive equilibrium $\mathcal{E}_1 = (Q^*,T^*,I^*,V^*,Z^*)$ in  system \eqref{1}.
\section{Stability analysis}
\subsection{Local stability}
In this subsection, we will examine the local stability of the equilibria  identified in the previous analysis. Our focus will be on the stability of the infection-free equilibrium $\mathcal{E}_0$. The Jacobian matrix of system $\eqref{1}$ evaluated at the infection-free equilibrium is given by
$$ J = 
\begin{pmatrix}
 -\mu_1-\sigma  & \varrho  & 0 & 0 & 0 \\
 \sigma  & -\mu_2-\varrho  & 0 & -\frac{\beta  \Lambda  \sigma }{\mu_1(\mu_2+\varrho )+\mu_2 \sigma } & 0
   \\
 0 & 0 & -\frac{b p}{\mu_5}-\mu_3 & \frac{\beta  \Lambda  \sigma }{\mu_1 (\mu_2+\varrho )+\mu_2
   \sigma } & 0 \\
 0 & 0 & \xi & -\mu_4 & 0 \\
 0 & 0 & \frac{b c}{\kappa  \mu_5} & 0 & -\mu_5 \\
\end{pmatrix}.
 $$
The eigenvalues of the Jacobian matrix can be calculated as
\begin{align*}
\lambda_1 &= -\mu_5, \\
\lambda_2 &= \frac{1}{2} \left(-\sqrt{2 \sigma (\mu_1-\mu_2+\varrho)+(-\mu_1+\mu_2+\varrho)^2+\sigma^2}-\mu_1-\mu_2-\sigma-\varrho\right), \\
\lambda_3 &= \frac{1}{2} \left(\sqrt{2 \sigma (\mu_1-\mu_2+\varrho)+(-\mu_1+\mu_2+\varrho)^2+\sigma^2}-\mu_1-\mu_2-\sigma-\varrho\right), \\
\lambda_4 &= -\frac{\frac{\sqrt{4 \beta\Lambda\mu_5^2\sigma\xi+(\mu_1(\mu_2+\varrho)+\mu_2\sigma)(bp+\mu_5(\mu_3-\mu_4))^2}}{\sqrt{\mu_1(\mu_2+\varrho)+\mu_2\sigma}}+bp+\mu_5(\mu_3+\mu_4)}{2\mu_5}, \\
\lambda_5 &= -\frac{-\frac{\sqrt{4 \beta\Lambda\mu_5^2\sigma\xi+(\mu_1(\mu_2+\varrho)+\mu_2\sigma)(bp+\mu_5(\mu_3-\mu_4))^2}}{\sqrt{\mu_1(\mu_2+\varrho)+\mu_2\sigma}}+bp+\mu_5(\mu_3+\mu_4)}{2\mu_5}.
\end{align*}
   It is noted that these eigenvalues hold significant information about the stability of the system, with negative eigenvalues indicating stability at the equilibrium.

Based on the expressions
\begin{multline*}
    (\mu_1+\mu_2+\sigma +\varrho )^2- \left( \sqrt{2 \sigma (\mu_1-\mu_2+\varrho )+(-\mu_1+\mu_2+\varrho )^2+\sigma ^2 } \right)^2\\ = 4 (\mu_1 (\mu_2+\varrho )+\mu_2 \sigma ) >0
    \end{multline*}
    and 
\begin{multline*}
   (b p+\mu_5 (\mu_3+\mu_4))^2-\left(\tfrac{\sqrt{4 \beta \Lambda \mu_5 ^2 \sigma \xi+(\mu_1 (\mu_2+\varrho )+\mu_2 \sigma ) (b p+\mu_5 (\mu_3-\mu_4))^2}}{\sqrt{\mu_1 (\mu_2+\varrho )+\mu_2 \sigma }} \right)^2\\
   = 4 \mu_4 \mu_5 (1-\mathcal{R}_{0}) (b p+\mu_3 \mu_5),
\end{multline*}
it is evident that all the eigenvalues are negative if the basic reproduction number $\mathcal{R}_{0}$ is less than 1. This observation implies that the infection-free equilibrium solution is asymptotically stable if $\mathcal{R}_{0} < 1$.
 
 We now analyze the stability of the infection equilibrium. The Jacobian matrix of system \eqref{1} at this equilibrium is given by
 $$ J_1 = 
\begin{pmatrix} -\mu_1-\sigma  & \varrho  & 0 & 0 & 0 \\
 \sigma  & -\mu_2-\varrho - \beta V^*  & 0 & - \beta T^* & 0
   \\
 0 & \beta V^* & -p Z^*-\mu_3 & \beta T^* & - p I^* \\
 0 & 0 & \xi & -\mu_4 & 0 \\
 0 & 0 & \frac{c Z^* \kappa}{(\kappa+I^*)^2} & 0 & \frac{c I^*}{\kappa+I^*}-\mu_5 \\
\end{pmatrix}.
 $$
 To simplify our analysis, we employ symbolic computation by reducing the number of parameters. At the equilibrium point $\mathcal{E}_1$, the relation $\mu_3 + p Z^* = \frac{\beta T^* \xi}{\mu_4}$ holds. Using this, we define the  transformations $\alpha_1 = \frac{\beta T^* \xi}{\mu_4}$, $\alpha_2 = \varrho + \mu_2 + \beta V^*$, $\alpha_3 = p I^*$, $\alpha_4 = \frac{c Z^* \kappa}{(\kappa + I^*)^2}$, $\alpha_5 = \mu_5 - \frac{c I^*}{\kappa + I^*}$, and $\alpha_6 = \beta V^*$ to obtain 
 $$ J_2 = 
\begin{pmatrix}
 -\mu_1-\sigma  & \varrho  & 0 & 0 & 0 \\
 \sigma  & - \alpha_2  & 0 & - \alpha_1 \frac{\mu_4}{\xi} & 0
   \\
 0 & \alpha_6 & - \alpha_1 & \alpha_1 \frac{\mu_4}{\xi} & - \alpha_3 \\
 0 & 0 & \xi & -\mu_4 & 0 \\
 0 & 0 & \alpha_4 & 0 & -\alpha_5 
\end{pmatrix}.
$$
  The characteristic equation of $J_2$ is given by 
  $$
\lambda^5+\xi_1 \lambda^4+ \xi_2 \lambda^3+\xi_3 \lambda^2+\xi_4 \lambda + \xi_5=0,
$$
where
$$
\begin{aligned}
\xi_1= {}& \alpha_5 + \mu_4 + \alpha_l + \alpha_2 + \sigma + \mu_1,
 \\
\xi_2= {}& \alpha_1 \alpha_2+\alpha_1 \alpha_5+\alpha_1 \mu_1+\alpha_1 \sigma+\alpha_2 \alpha_5 
+\alpha_2 \mu_1+\alpha_2 \mu_4+\alpha_3 \alpha_4+\alpha_5 \mu_1 \\ 
&+\alpha_5 \mu_4+\alpha_5 \sigma+\mu_1 \mu_4+\mu_4 \sigma+\sigma \left(\alpha_2- \varrho \right), \\
\xi_3= {}& \alpha_1 \alpha_2 \alpha_5 + \alpha_1 \alpha_2 \mu_1  + \alpha_1 \alpha_5 \mu_1 + \alpha_1 \alpha_5 \sigma + \alpha_1 \alpha_6 \mu_4 + \alpha_2 \alpha_3 \alpha_4 +\alpha_2 \alpha_5 \mu_1 \\
&+ \alpha_2 \alpha_5 \mu_4 +  \alpha_2 \mu_1 \mu_4 + \alpha_3 \alpha_4 \mu_1 + \alpha_3 \alpha_4 \mu_4 + \alpha_3 \alpha_4 \sigma + \alpha_5 \mu_1 \mu_4 + \alpha_5 \mu_4 \sigma \\&+ \sigma \left(\alpha_1+\alpha_4+\alpha_4 \right) \left( \alpha_2-\varrho \right) ,\\
\xi_4 = {}&  \alpha_1 \alpha_2 \alpha_5 \mu_1 + \alpha_1 \alpha_5 \alpha_6 \mu_4 + \alpha_1 \alpha_6 \mu_1 \mu_4 + \alpha_1 \alpha_6 \mu_4 \sigma + \alpha_2 \alpha_3 \alpha_4 \mu_1 + \alpha_2 \alpha_3 \alpha_4 \mu_4 \\
&+ \alpha_2 \alpha_5 \mu_1 \mu_4 + \alpha_3 \alpha_4 \mu_1 \mu_4 +\alpha_3 \alpha_4 \mu_4 \sigma + \sigma \left(\alpha_1 \alpha_5+\alpha_3 \alpha_4+ \mu_4\alpha_5 \right) \left( \alpha_2-\varrho \right),\\
\xi_5 ={} & \alpha_1 \alpha_5 \alpha_6 \mu_1 \mu_4 + \alpha_1 \alpha_5 \alpha_6 \mu_4 \sigma + \alpha_2 \alpha_3 \alpha_4 \mu_1 \mu_4 + \sigma \alpha_3 \alpha_4 \mu_4 \left(\alpha_2 - \varrho \right).
\end{aligned}
$$
Since $\alpha_2 - \varrho > 0$, it follows that $\xi_2, \xi_3, \xi_4, \xi_5 > 0$.

 Algebraic calculations show that
$$
\xi_1 \xi_2 - \xi_3 = \Phi_1 + \sigma (\alpha_2 + \mu_1 + \sigma)(\alpha_2 - \varrho) + \alpha_1 \mu_4 (\alpha_2 - \alpha_6).
$$
Furthermore,
$$
\begin{aligned}
\MoveEqLeft[2] \xi_3-\frac{\xi_l \left(\xi_5-\xi_1  \xi_4\right)}{\xi_3-\xi_1 \cdot \xi_2} \\={} & \frac{1}{\Phi_2+ \alpha_1 \alpha_4 \left(\alpha_2-\alpha_6\right)+\sigma \left(\alpha_2+\mu_1+\sigma\right)\left(\alpha_2-\varrho\right)} \\ & \times
(\Phi_3 + \left(\alpha_2-\alpha_6\right) (\alpha_1^2\alpha_2\alpha_5\mu_4+\alpha_1^2\alpha_2\mu_1\mu_4+\alpha_1^2\alpha_5^2\mu_4 \\ &+2\alpha_1^2\alpha_5\mu_1\mu_4+2\alpha_1^2\alpha_5\mu_4\sigma +\alpha_1^2\alpha_6\mu_4+\alpha_1^2 \mu_1^2 \mu_4 + \Phi_4 \\
& + \sigma\left(\alpha_2-\varrho\right) (\alpha_1^3 \alpha_2 + \alpha_1^3\mu_1+\alpha_1^3\sigma+\alpha_1^2\alpha_2^2+\alpha_1^2\alpha_2\alpha_5+3\alpha_1^2\alpha_2\mu_4+2 \alpha_1^2\alpha_2\sigma\\&+\alpha_1^2\alpha_5\mu_1\sigma  + \alpha_1^2\alpha_5\sigma+\alpha_1^2\mu_1^2+3 \alpha_1^2\mu_1\mu_4+2\alpha_1^2\mu_1\sigma+3\alpha_1^2\mu_4\sigma+ \Phi_5 )).
\end{aligned}
$$
 Detailed but lengthy calculations (see Appendix) confirm the positivity of $\Phi_1$, $\Phi_2$, $\Phi_3$, $\Phi_4$  and $\Phi_5$. Although the expressions are extensive, we have thoroughly verified that they are positive. Moreover, since $\alpha_2 - \alpha_6 = \varrho + \mu_2 > 0$, we conclude that $\xi_1 \xi_2 - \xi_3 > 0$ and $\xi_3-\frac{\xi_l \left(\xi_5-\xi_1  \xi_4\right)}{\xi_3-\xi_1 \cdot \xi_2} > 0 $. \\$ $\\
Applying the Routh--Hurwitz criterion \cite{Bodson}, which asserts that the number of roots of the characteristic polynomial with positive real parts (right half-plane roots) corresponds to the number of sign changes in the first column of the Routh array, we find that the first column of the complete Routh scheme is entirely positive. This indicates that all roots of the characteristic polynomial have negative real parts. Consequently, the infection equilibrium point $\mathcal{E}_1$ is locally asymptotically stable.
\subsubsection{Global stability of the infection-free equilibrium} In order to analyze the global stability of the infection-free equilibrium, we will use a  matrix-theoretic method defined in a previous study by Shuai et al.~\cite{shuai}. To apply this method, we need to calculate the matrices $F$ and $V_1$, which are given in Section~4. Using these matrices, we can compute the following matrix:
\begin{equation*}
V_1^{-1}=\begin{pmatrix}
 \frac{\mu_5}{b p+\mu_3 \mu_5} & 0 \\
 \frac{\mu_5 \xi}{b \mu_4 p+\mu_3 \mu_4 \mu_5} & \frac{1}{\mu_4} \\
\end{pmatrix}.
\end{equation*}
We will follow the notation introduced in Theorem 2.1 of \cite{shuai}. This notation involves defining $x=(I,V)$ and $y=(Q,T,Z)$ for the model \eqref{1} we have 
\begin{equation*}
f(x, y):=\left(\begin{array}{c}
I p\left(Z-\frac{b}{\mu_{\mathrm{s}}}\right)+\beta V\left(\frac{\Lambda \sigma}{\mu_1\left(\mu_2+e\right)+\mu_2 \sigma}-T\right)
\end{array}\right) .
\end{equation*}
By adopting this notation, we proceed by applying Theorem 2.1, which states that if \(f(x, y) \geq 0\) in a subset \(\mathcal{D} \subset \mathbb{R}_{+}^{5}\), along with the conditions \(F \geq 0\), \(V_1^{-1} \geq 0\), and \(\mathcal{R}_0 \leq 1\), then the function \(Q = \omega^T V^{-1} x\) can be used as a Lyapunov function for the model. In our case, we apply the theorem by defining $Q=\omega^T V_1^{-1} x$, choosing $\omega^T$ as the left eigenvector of the matrix $V_1^{-1} F$ that corresponds to the basic reproduction number $\mathcal{R}_0$, and differentiating $Q$ along solutions of \eqref{1} to obtain
$$
\begin{aligned}
Q^{\prime} & =\omega^T V_1^{-1} x^{\prime}=\omega^T V_1^{-1}(F-V) x-\omega^T V_1^{-1} f(x, y) \\
& =\left(\mathcal{R}_0-1\right) \omega^T x-\omega^T V_1^{-1} f(x, y) .
\end{aligned}
$$
As seen above, the conditions $F>0$ and $V_1^{-1}>0$ are always true, but the condition $f(x, y)\geq 0$ does not always hold. It is easy to see that the first  coordinate is positive if the conditions
$$ Z(t) > \frac{b}{\mu_5} \quad  \text{and}\quad T (t)< \frac{ \Lambda \sigma}{\mu_1 (\mu_2+\varrho)+\mu_2 \sigma }$$
hold.
For the first  condition, we have 
$$ \frac{d Z}{d t}\geq b-\mu_5 Z.$$ A simple comparison principle shows that $ Z(t) \geq \frac{b}{\mu_5}$ for $t$ large enough.

 For the second condition, we consider the infection-free subsystem 
\begin{align*}
\frac{dT}{dt}&=\Lambda+\varrho T-\sigma Q-\mu_1 Q, \\ 
\frac{dQ}{dt}&=\sigma Q-\varrho T-\mu_2 T.
\end{align*}
We will show that all solutions of this subsystem tend to the equilibrium $(Q^*,T^*)$, where 
$$Q^*\coloneqq \frac{\Lambda  (\mu_2+\varrho )}{\mu_1 (\mu_2+\varrho )+\mu_2 \sigma }\quad \mbox{and}\quad T^*\coloneqq\frac{\Lambda  \sigma }{\mu_1 (\mu_2+\varrho )+\mu_2 \sigma }.$$
To do so, we use the Dulac--Bendixson criterion. This criterion states that if the divergence of a continuously differentiable vector field on a simply connected region does not change sign and is nonzero, then the system cannot admit periodic orbits within that region. In our case, we introduce the positive function $g(Q,T) = \frac{1}{T}$  defined in the region $T > 0$ and $Q > 0$. We  compute the divergence of the vector field multiplied by $g$ as
$$
\operatorname{div}(g \cdot \vec{F}) = -\frac{Q \sigma +T (\mu_1+\sigma )}{T^2}
$$
This expression is negative for all $(Q,T)$ in the phase space. Therefore, by the Dulac--Bendixson criterion, there are no closed orbits and all  positive solutions tend to the unique equilibrium $(Q^*,T^*)$.  Again, using a simple comparison argument, we obtain that $T<T^*$ for $t$ large enough.
\subsection{Uniform persistence}
In this section, we will show that the infection related compartments -- the infected cells  $I(t)$ and the virus particles $V(t)$ -- will persist if $\mathcal{R}_0>1$. In order to state our main result on uniform persistence of $I(t)$ and $V(t)$, we will recall some theory from \cite{smith}.

\begin{definition} Let $X$ be a nonempty set and $\rho\colon X\to \mathbb{R}_+$. 
A semiflow $\Phi\colon  \mathbb{R}_+\times X\to X$  is called \emph{uniformly weakly $\rho$-persistent}, if there exists some $\varepsilon>0$ such that $$\limsup_{t\to\infty}\rho(\Phi(t,x))>\varepsilon\qquad \forall x\in X, \ \rho(x)>0.$$
$\Phi$ is called \emph{uniformly (strongly) $\rho$-persistent} if  there exists some $\varepsilon>0$ such that $$\liminf_{t\to\infty}\rho(\Phi(t,x))>\varepsilon\qquad \forall x\in X, \ \rho(x)>0.$$

A set $M\subseteq X$ is called \emph{weakly $\rho$-repelling} if there is no $x\in X$ such that $\rho(x)>0$ and $\Phi(t,x)\to M$ as $t\to\infty$. 
\end{definition}

System \eqref{1} generates a continuous flow on the  state space $$X:=\{(Q,T,I,V,Z) \in \mathbb{R}_+^5\}\subset \mathbb{R}_+^5 .$$
To keep our notations simple,  we will aplly the notation $x=(Q,T,I,V,Z)\in X$ for the state of the system.  As usual, the $\omega$-limit set of a point $x\in X$ is defined as 
$$\omega(x)\coloneqq\{y \in X: \exists  \{t_n\}_{n \geq 1} \text{ such that } \ t_n \to \infty,  \Phi(t_n,x) \to y \text{ as } n \to \infty\} .$$

\begin{theorem}
$I(t)$ and $V(t)$ are uniformly persistent if $\mathcal{R}_0>1$.
\end{theorem}
\begin{proof}
    Suppose $\mathcal{R}_0>1$ and choose $\rho(x)=I+\frac{\mu_3}{\xi}V$. Define  the infection-free subspace $$X_0 \coloneqq\{x \in X: \rho(x)=0\}=\{(Q,T,0,0,Z)\in \mathbb{R}_+^5\}.$$
It is clear that $X_0$ is invariant and that $\Omega\coloneqq\cup_{x\in X_0}=\{E_0\}$. Following \cite[Chapter~8]{smith}, we first prove weak $\rho$-persistence. Define $M_1=\{E_0\}$, then clearly $M_1\subset\Omega$, it is isolated, compact, invariant and acyclic. To complete the proof of weak $\rho$-persistence, applying \cite[Theorem 8.17]{smith}, we need to show that $M_1$ is weakly $\rho$-repelling. Suppose the contrary, then there exists  solution with its $\omega$-limit set being $M_1$, such that $I+k V>0$ with a constant $k$ to be determined later. From this convergence, we have that for $t$ large enough,
\begin{align*}
    T(t)> \frac{\Lambda  \sigma}{\mu_1 (\mu_2+\varrho )+\mu_2 \sigma }-\varepsilon
    \quad \mbox{and}\quad
    Z(t)<\frac{b}{\mu_5}+\varepsilon
\end{align*}
hold. Then, for $t$ large enough, we can estimate $(I(t)+k V(t))'$ as
\begin{align}\label{pers}
    \left(I(t)+k V(t)\right)'={}&\beta T(t)V(t)-\mu_3 I(t)-p  I(t)Z(t)+k\left(\xi I(t)-\mu_4V(t)\right)\nonumber\\
    \geq{}& \left(\frac{\beta\Lambda\sigma}{\mu_1 (\mu_2+\varrho )+\mu_2 \sigma} -\beta\varepsilon\right)V(t)-\left(\mu_3+p\frac{b}{\mu_5}+p\varepsilon\right)I(t)\nonumber\\
   & +k(\xi I(t)-\mu_4 V(t))\nonumber\\
   ={}&K(I(t)+kV(t))
    \end{align}
for appropriately chosen positive constants $K$ and $k$. From the above calculations we obtain that finding such positive constants $K$ and $k$, enables us to estimate the solution from below by the solution of the equation 
$$(I(t)+kV(t))'=K(I(t)+kV(t)),$$
which would contradict $(I(t)+kV(t))\to 0$.

Finding such constants $K$ and $k$ is equivalent to finding a positive eigenvalue $K$ with positive corresponding eigenvector of the matrix
\begin{equation*}
    \begin{pmatrix}
        -\mu_3-\frac{bp}{\mu_5}-p\varepsilon&\frac{\beta\Lambda\sigma}{\mu_1(\mu_2+\varrho)+\mu_2\sigma}-\beta\varepsilon\\
        \xi& -\mu_4
    \end{pmatrix}.
\end{equation*}
As $\varepsilon$ can be chosen arbitrarily small, by continuity, it is enough to find a positive eigenvalue with corresponding positive eigenvector of the matrix 
\begin{equation*}
    M=\begin{pmatrix}
        -\mu_3-\frac{bp}{\mu_5}&\frac{\beta\Lambda\sigma}{\mu_1(\mu_2+\varrho)+\mu_2\sigma}\\
        \xi& -\mu_4
    \end{pmatrix}.
\end{equation*}
As the off-diagonal elements of $M$ are nonnegative, $M$ is a Metzler matrix, so we can apply \cite[Theorem 11]{luenberger}, which states that for any Metzler matrix $A\in\mathbb{R}^{n\times n}$, the spectral abscissa $\eta( A )$ of A (i.e., the maximum of the real parts of the eigenvalues of $A$)  is an eigenvalue of $A$
and there exists a nonnegative eigenvector $x \geq 0$, $x \neq 0$ such that $Ax =\eta(A ) x$. Hence, we only need to show that $\eta(M)$ is positive if $\mathcal{R}_0>1$. The characteristic polynomial of $M$ takes the form 
$$\lambda ^2+\lambda  \left(\frac{b p}{{\mu_5}}+{\mu_3}+{\mu_4}\right) +\frac{b {\mu_4} p}{{\mu_5}}+{\mu_3} {\mu_4}-\frac{\beta  \Lambda  \xi  \sigma }{{\mu_1} (\mu_2+\rho )+{\mu_2} \sigma }.$$
The first and second coefficients of the characteristic polynomial are positive, while the constant term is negative if and only if $\mathcal{R}_0 > 1$. This is demonstrated by the following proof:
\begin{align*}
    \mathcal{R}_{0} > 1 &\Leftrightarrow \frac{\beta \Lambda \mu_5 \sigma \xi}{\mu_4 (b p + \mu_3 \mu_5) (\mu_1 (\mu_2 + \varrho) + \mu_2 \sigma)} > 1 \\&\Leftrightarrow \frac{\beta \Lambda \sigma \xi}{(\mu_1 (\mu_2 + \varrho) + \mu_2 \sigma)} > \frac{\mu_4 b p + \mu_3 \mu_4 \mu_5}{\mu_5}
\\&\Leftrightarrow \frac{b \mu_4 p}{\mu_5} + \mu_3 \mu_4 - \frac{\beta \Lambda \sigma \xi}{(\mu_1 (\mu_2 + \varrho) + \mu_2 \sigma)} < 0.
\end{align*}
Therefore, if this condition holds, then, according to the Routh--Hurwith theorem, there exists an eigenvalue with positive real part. This implies  the positivity of the spectral
abscissa $\eta( M )$ and the corresponding eigenvector, hence, for $\varepsilon$ sufficiently small (i.e., for  $t$ sufficiently large) we can find  positive constants $K$ and $k$ such that the equality \eqref{pers} holds, which implies the weak persistence of $I(t) + k V(t)$ in the case $\mathcal{R}_0 > 1$. Applying \cite[Theorem 4.5]{smith}, the uniform persistence of   $I(t) + k V(t)$ follows. Simple calculations yield the uniform persistence of $I(t)$ and $V(t)$.
\end{proof}

\subsection{Transcritical bifurcation at \texorpdfstring{$\mathcal{R}_0 =1$}{R0=1}}
In the following, we use the centre manifold theory \cite{Castillo} to explore the possibility of transcritical bifurcation in \eqref{1}. To do so, a bifurcation parameter $\beta^*$ is chosen, by solving for $\beta$ from $\mathcal{R}_0 =1$, giving
$$ \beta^* = \frac{\mu_4 (b p+\mu_3 \mu_5) (\mu_1 (\mu_2+\varrho )+\mu_2 \sigma )}{\Lambda  \mu_5 \sigma  \xi}.$$
The matrix $J_1(\mathcal{E}_0, \beta^*)$ (which is equal to $J_1$ at $\mathcal{E}_0$ with $\beta=\beta^*$) has one simple zero eigenvalue and four negative eigenvalues. 

To calculate the following formulas we introduce the notation $x = (Q,T,I,V,Z)$, 
$$\begin{aligned}
a & =\sum_{k, i, j=1}^5 v_k w_i w_j \frac{\partial^2 f_k}{\partial x_i \partial x_j}(\mathcal{E}_0, \beta^*), \\
b & =\sum_{k, i=1}^5 v_k w_i \frac{\partial^2 f_k}{\partial x_i \partial \beta}(\mathcal{E}_0, \beta^*),
\end{aligned} $$
where $w, v$ are the right and left eigenvectors of $J$ corresponding to the zero eigenvalue defined as follow: 
$$ w = \left\{-\frac{\kappa  \mu_5 \varrho  (b p+\mu_3 \mu_5)}{b c (\mu_1 (\mu_2+\varrho )+\mu_2
   \sigma )},-\frac{\kappa  \mu_5 (\mu_1+\sigma ) (b p+\mu_3 \mu_5)}{b c (\mu_1 (\mu_2+\varrho )+\mu_2 \sigma )},\frac{\kappa  \mu_5^2}{b c},\frac{\kappa  \mu_5^2 \xi}{b c \mu_4},1\right\}.$$
   $$v = \left\{0,0,\frac{\mu_5 \xi}{b p+\mu_3 \mu_5},1,0\right\}.$$
  As $v_1=v_2=v_5=0$ the derivatives of $f_1$, $f_2$ and $f_5$ are not needed. All second order derivatives of $f_3$ and $f_4$ are zero, except for
$$
\begin{aligned}
\frac{\partial^2 f_3}{\partial x_2 \partial x_4} & = \beta^*, & \frac{\partial^2 f_3}{\partial x_3 \partial x_5} & =- p, \\
  \frac{\partial^2 f_2}{\partial x_4 \partial \beta} & = \frac{\Lambda  \sigma }{\mu_1 (\mu_2+\varrho )+\mu_2 \sigma } .
\end{aligned}
$$ 
Therefore, the quantities $a$ and $b$ are given by
$$ a = -2 \left( \frac{ \kappa ^2 \mu_5 ^3 \xi (\mu_1+\sigma)}{(b c)^2 \Lambda \sigma}+ \frac{p \mu_5^3 \xi \kappa}{b c (b p + \mu_3 \mu_5)} \right) $$
$$b = \frac{\kappa \mu_5^3 \xi^2 }{ (b p + \mu_3 \mu_5) bc \mu_4 } \frac{\Lambda  \sigma }{\mu_1 (\mu_2+\varrho )+\mu_2 \sigma } $$
Based on these results, we can derive the following theorem.
\begin{theorem} 
A transcritical bifurcation of forward-type occurs at $\mathcal{R}_0 =1$
\end{theorem}
To further illustrate the bifurcation dynamics, we present two figures. The first figure \ref{fig:bif1} demonstrates how the viral load, $V^*$, changes as the transmission rate, $\beta$, increases. The second figure \ref{fig:bif2} shows the behavior of the infected cell population, $I^*$, as a function of $\beta$, incorporating two scenarios: one where the natural death rate of CTL cells, $\mu_5$, exceeds the critical value $c$, and another where $\mu_5$ is below $c$. These plots highlight how $I^*$ varies with increasing $\beta$, clearly demonstrating the occurrence of a transcritical bifurcation at the threshold.
\begin{figure}[htbp]
    \centering
    \includegraphics[width=0.6\linewidth]{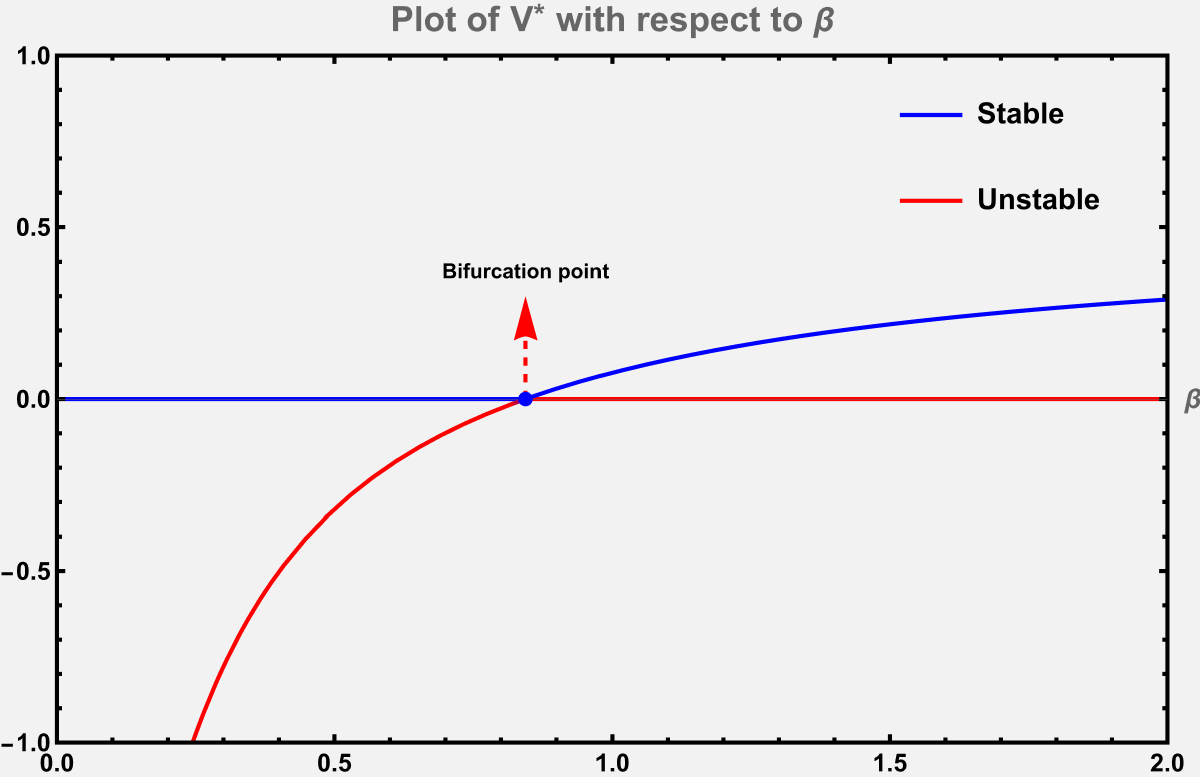}
    \caption{Viral load $V^*$ as a function of the transmission rate $\beta$, illustrating the system's behavior near the bifurcation point. The plot demonstrates how the viral load increases as $\beta$ crosses the critical threshold.}
    \label{fig:bif1}
\end{figure}
\begin{figure}[htbp]
  \begin{minipage}[t]{0.5\textwidth}
    \centering
    \includegraphics[width=\textwidth]{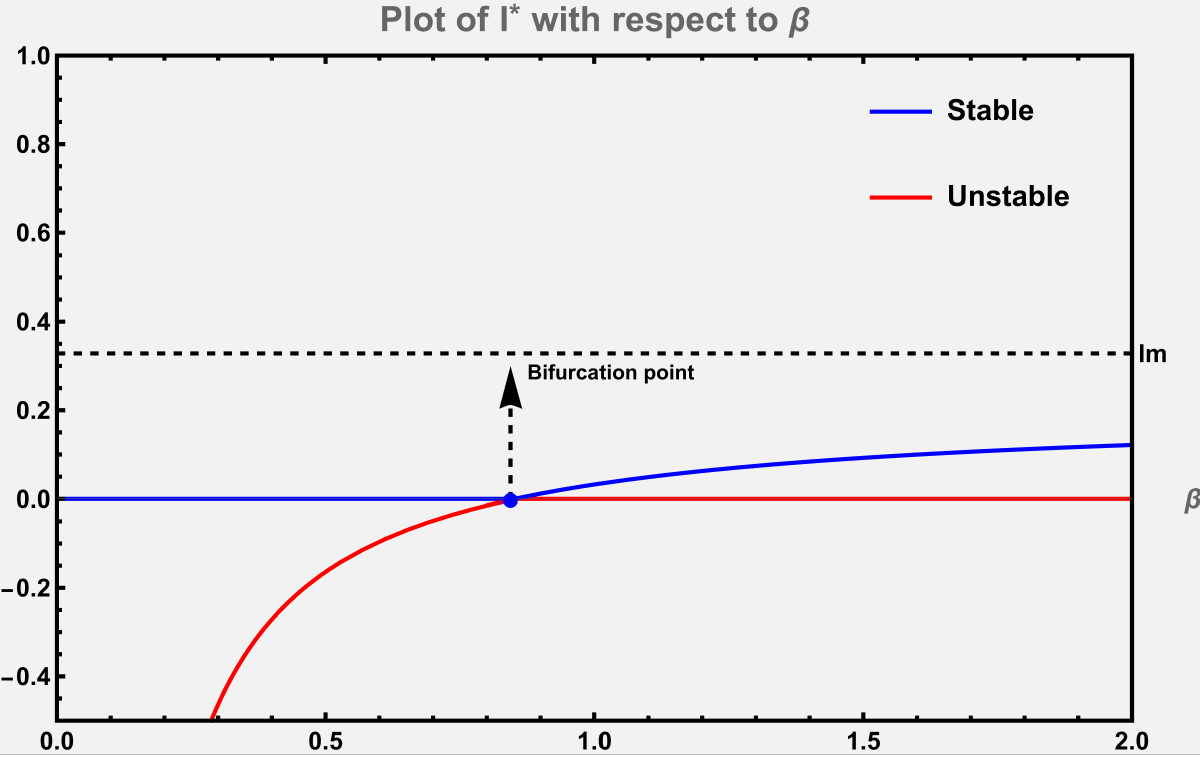} 
  \end{minipage}%
  \begin{minipage}[t]{0.5\textwidth}
    \centering
    \includegraphics[width=\textwidth]{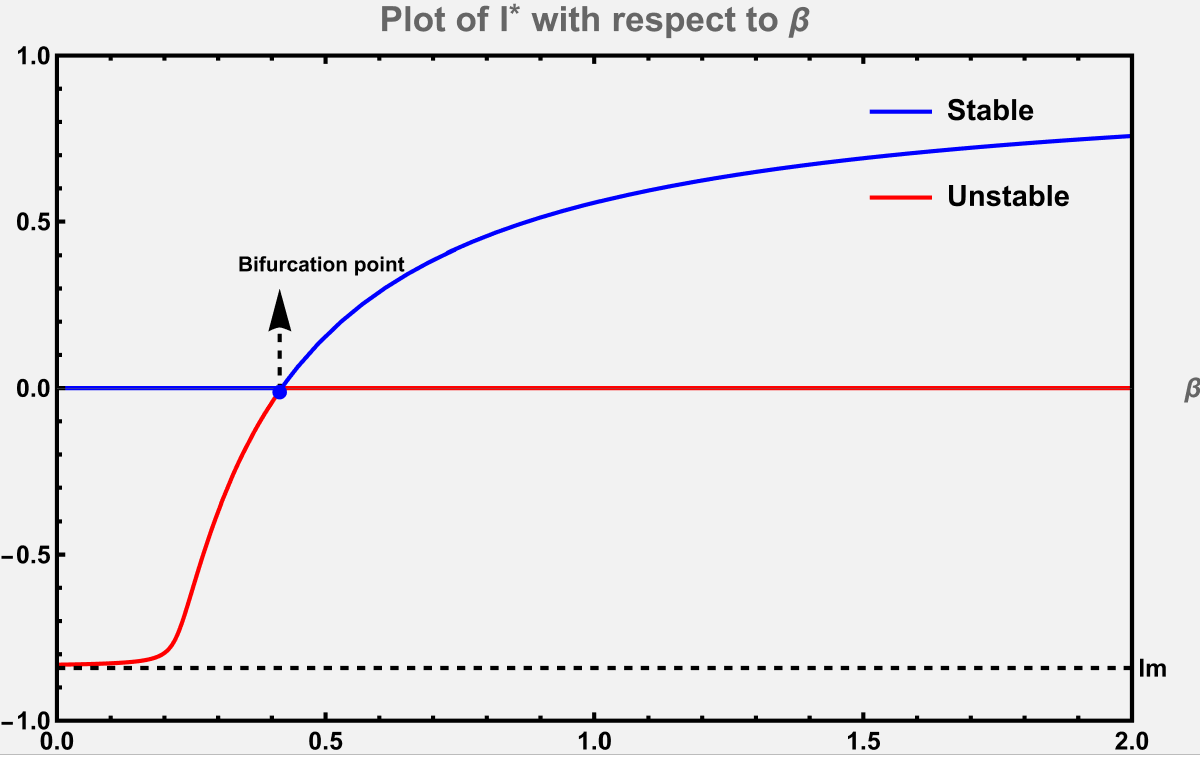} 
  \end{minipage}
  \caption{Infected cell population $I^*$ as a function of the transmission rate $\beta$. The figure shows two cases: one where the natural death rate of CTL cells $\mu_5$ exceeds the critical value $c$, and another where $\mu_5$ is below $c$. The transcritical bifurcation is visible as $\beta$ increases.}
  \label{fig:bif2}
\end{figure}
\section{Numerical simulation}
\subsection{Time series analysis}
 In this section, we present time series analyses for two scenarios: $\mathcal{R}_0 < 1$ and $\mathcal{R}_0 > 1$. The following plots show the time series of  \eqref{1} for these two scenarios.
 \begin{center}
\begin{table}[h!]%
\centering
\caption{Model parameters  and their values in Figures \ref{fig:Num1} and \ref{fig:Num2}.\label{tab3}}%
\begin{tabular*}{240pt}{@{\extracolsep\fill}lcccc@{\extracolsep\fill}}
\toprule
\textbf{Parameter}  & \textbf{Value for Fig.~\ref{fig:Num1}} & \textbf{Value for Fig.~\ref{fig:Num2}}  \\
\midrule
$\Lambda$    & 100 & 0.1     \\
$\varrho$   & 0.1 & 0.1     \\
$ \sigma$   & 0.2 & 0.2   \\
 $\mu_1$  & 0.01 & 0.01  \\
$ \beta$  & 0.3 & 0.3 \\
$ \mu_2$  & 0.01 & 0.01 \\
$\mu_3$  & 0.1 & 0.1 \\
$p$   & 0.05 & 0.05 \\
$\xi$  & 10 & 0.01 \\
$\mu_4$  & 0.1 & 0.1 \\
$b$  & 1 & 1 \\
$c$  & 0.1 & 0.1 \\
$\kappa$  & 50 & 50 \\
$\mu_5$  & 0.05 & 0.05 \\
$Q(0)$   & 10 & 10 \\
$T(0)$   & 1 & 1 \\
$I(0)$   & 10 & 10 \\ 
$V(0)$   & 0 & 0 \\
$Z(0)$   & 10 & 10 \\
\bottomrule
\end{tabular*}
\end{table}
\end{center}
\begin{figure}[htbp]
  \centering
  \includegraphics[width=\textwidth]{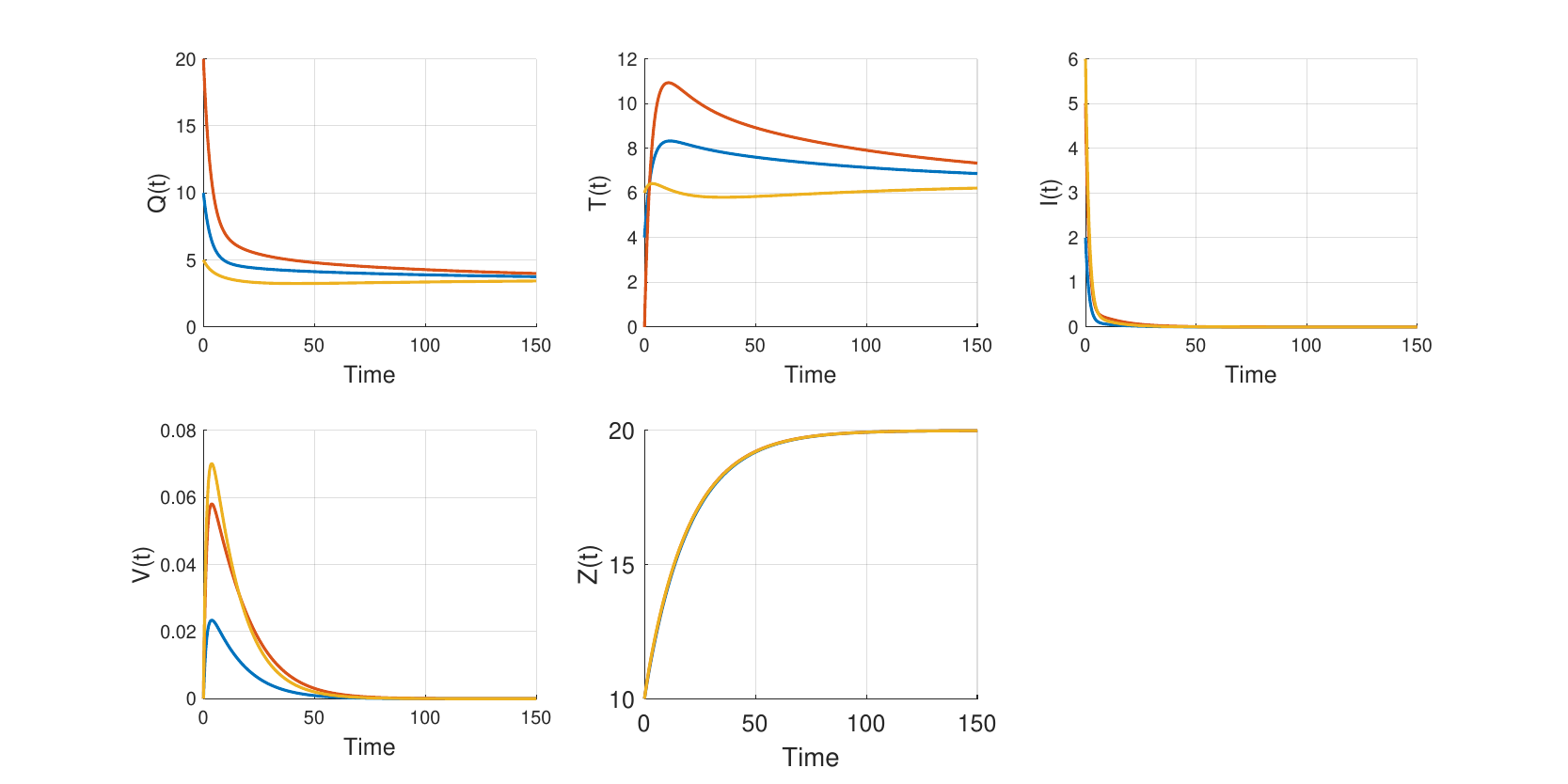}
  \caption{Time series plot for $\mathcal{R}_0 < 1$ scenario.}
  \label{fig:Num1}
\end{figure}

\begin{figure}[htbp]
  \centering
  \includegraphics[width=\textwidth]{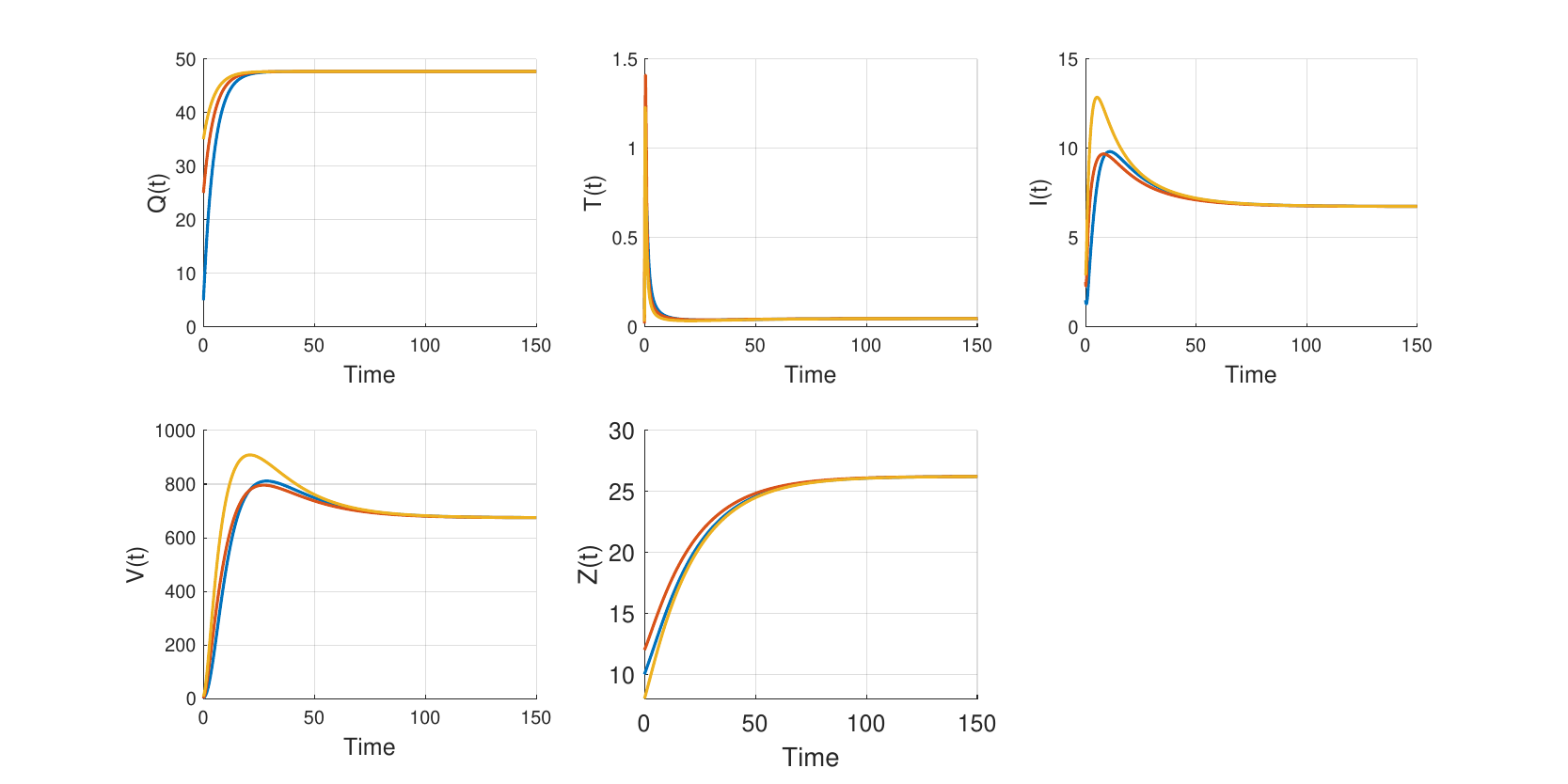}
  \caption{Time series plot for $\mathcal{R}_0 > 1$ scenario.}
  \label{fig:Num2}
\end{figure}
According to Figures \ref{fig:Num1} and \ref{fig:Num2}, when the basic reproductive number, $\mathcal{R}_0$, is less than 1, and when $\mathcal{R}_0$ is greater than 1. In the scenario where $\mathcal{R}_0$ is less than 1, we observed a distinct pattern in the time series analysis. The quiescent cells exhibit a decreasing trend, eventually converging to a stable value. Simultaneously, the number of healthy activated cells, representing cells that respond to the viral infection, increases and converges to a steady state. The infected cells gradually decrease and eventually diminish to zero, while the free virus particles initially increase, reach a peak, and then decrease, eventually stabilizing. Additionally, the immune response, represented by the CTL (cytotoxic T-lymphocyte) cells, shows an increasing trend throughout the observation period.\\$ $\\
On the other hand, when $\mathcal{R}_0$ is greater than 1, the time series analysis reveals a different behavior. Quiescent cells exhibit an increasing trend, suggesting a larger pool of inactive cells. In contrast, the number of healthy activated cells, representing the target cells for viral infection, experiences a decreasing trend. The infected cells initially increase, reach a peak, and then gradually converge to a stable value. Similarly, the free virus particles display an increasing trend, reaching a peak value, and subsequently stabilizing. Notably, the immune response, characterized by the CTL cells, exhibits a consistent increase throughout the observation period, indicating an intensified effort to combat the infection and control the spread of the virus.
\subsection{Analysis of the reproduction number's sensitivity}

In this subsection, we carry out a sensitivity analysis to explore how different parameters impact the basic reproduction number, $\mathcal{R}_0$. Using Partial Rank Correlation Coefficients (PRCC) analysis, we can evaluate the influence of each parameter on $\mathcal{R}_0$. Parameters with positive PRCC values show a direct relationship, where increasing them raises $\mathcal{R}_0$, while those with negative PRCC values are inversely related, meaning that an increase in their values lowers the reproduction number. According to our results, displayed in Figure~\ref{fig:PRCC}, the natural death rate of CTL cells ($\mu_5$) exerts the strongest positive effect on $\mathcal{R}_0$, followed by the transition rate of healthy activated cells to quiescent cells ($\varrho$). Meanwhile, the infection rate of healthy activated cells by free virus ($\beta$) and the production rate of CTL cells ($b$) were found to most effectively reduce $\mathcal{R}_0$.
 \begin{figure}[htbp]
 
{\hspace{-12pt}\includegraphics[width=0.97\textwidth]{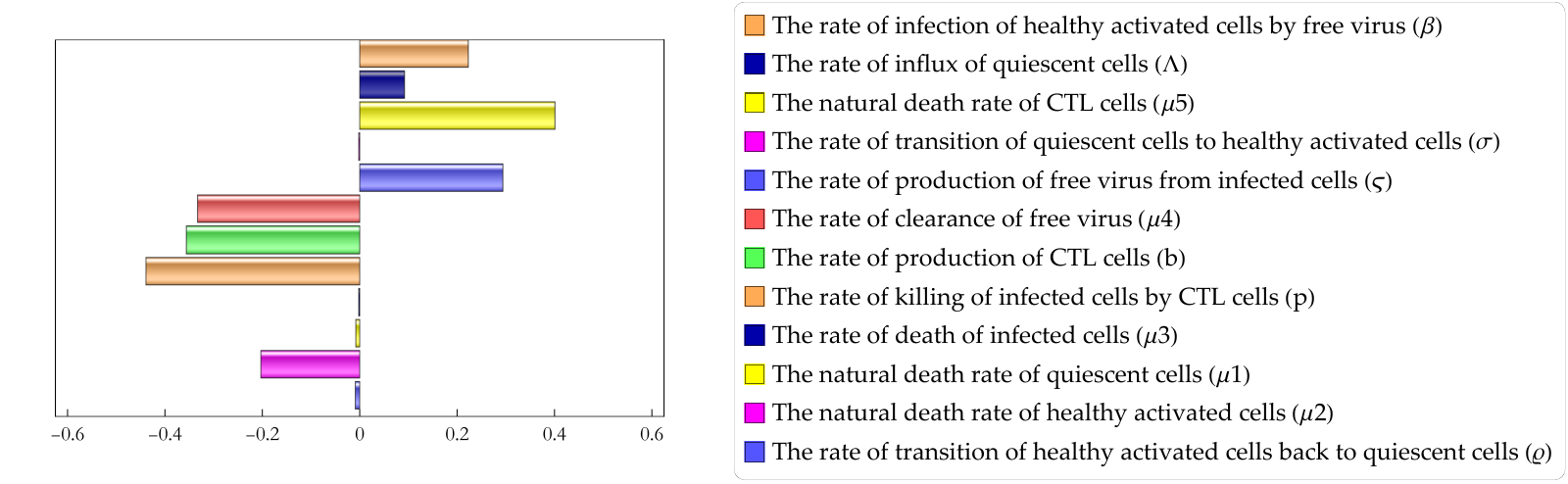}}
    \caption{Partial rank correlation coefficients of parameters of model \eqref{1}.}
    \label{fig:PRCC}
\end{figure} 
\section{Discussion and conclusion }
In this paper, we presented a comprehensive analysis of an HIV infection model that incorporates quiescent cells and the host immune response. Our model, based on a system of ordinary differential equations, captures the complex dynamics of viral infection within the host, offering insights into the equilibrium states, stability, and bifurcation phenomena. Through rigorous mathematical analysis, we explored both the infection-free and infection equilibria, providing a detailed understanding of how these equilibria govern the long-term behavior of the system.

The infection-free equilibrium represents the scenario where the virus is eradicated, and our results demonstrated that this equilibrium is both locally and globally stable when the basic reproduction number $\mathcal{R}_0$ is less than one. This implies that under appropriate conditions, the virus cannot persist in the host population. Conversely, when $\mathcal{R}_0$ exceeds one, the system tends toward the infection equilibrium, which we found to be locally asymptotically stable. This indicates that in cases where the virus establishes itself, the infection will persist over time unless effective interventions are implemented.

Moreover, we identified a transcritical bifurcation at the critical threshold $\mathcal{R}_0 = 1$. This bifurcation signifies a qualitative change in the system's dynamics, where the stability of the infection-free equilibrium is lost, and the infection equilibrium becomes stable. This critical point plays a crucial role in determining the system's response to changes in parameters, highlighting the importance of maintaining control over the reproduction number to prevent disease outbreaks.

In addition to the equilibrium and stability analysis, we performed a sensitivity analysis, showed that $\mu_5$ and $\varrho$ increase $\mathcal{R}_0$, while $\beta$ and $b$ decrease it.

Overall, our model provides valuable insights into the interactions between HIV, quiescent cells, and the immune system, shedding light on the factors that influence the persistence or eradication of the virus. The identification of the bifurcation point and the sensitivity analysis provide useful information for designing interventions aimed at controlling the infection. 

In conclusion, this research has broadened our understanding of HIV infection dynamics by incorporating quiescent cells and immune response mechanisms into the model. Our findings reveal critical thresholds and parameter sensitivities that can guide the development of more effective treatment strategies.    
\newpage
\appendix
\section*{Appendix}
\allowdisplaybreaks
\allowdisplaybreaks
\begin{align*}
 \MoveEqLeft[3] \xi_3-\frac{\xi_l \left(\xi_5-\xi_1  \xi_4\right)}{\xi_3-\xi_1 \cdot \xi_2} = \\  & \bigg(\big((\alpha_5+\mu_1+\sigma ) \alpha_2^2+\big((\alpha_5+\mu_1)^2+2 \sigma  (\alpha_5+\mu_1)+\sigma ^2+\alpha_6 \mu_4-\rho  \sigma \big) \alpha_2\\
 &+(\mu_1+\sigma ) \big(\alpha_5^2+(\mu_1+\sigma ) \alpha_5-\rho  \sigma \big)\big) \alpha_1^3+\big((\alpha_5+\mu_1+\sigma ) \alpha_2^3\\
 &+\big(2 \alpha_5^2+4 \mu_1 \alpha_5+3 \mu_4 \alpha_5+2 \mu_1^2+2 \sigma ^2+\alpha_3 \alpha_4+\alpha_6 \mu_4+3 \mu_1 \mu_4+(4 \alpha_5+4 \mu_1+3 \mu_4-\rho ) \sigma \big) \alpha_2^2\\
 &+\big(\alpha_5^3+4 \mu_1 \alpha_5^2+3 \mu_4 \alpha_5^2+4 \mu_1^2 \alpha_5-\alpha_6 \mu_4 \alpha_5+6 \mu_1 \mu_4 \alpha_5+\mu_1^3+\sigma ^3+2 \alpha_6 \mu_4^2\\
 &+(4 \alpha_5+3 (\mu_1+\mu_4)-2 \rho ) \sigma ^2+3 \mu_1^2 \mu_4-\alpha_6 \mu_1 \mu_4+\big(4 \alpha_5^2+8 \mu_1 \alpha_5+6 \mu_4 \alpha_5+3 \mu_1^2-\alpha_6 \mu_4\\
 &+6 \mu_1 \mu_4-(\alpha_5+2 \mu_1+3 \mu_4) \rho \big) \sigma +2 \alpha_3 \alpha_4 (\alpha_5+\mu_1+\sigma )\big) \alpha_2+\alpha_5 \mu_1^3+(\alpha_5-\rho ) \sigma ^3+2 \alpha_5^2 \mu_1^2\\
 &-\alpha_6^2 \mu_4^2+\big(2 \alpha_5^2+3 (\mu_1+\mu_4) \alpha_5-\rho  \alpha_5-2 \mu_1 \rho -\mu_4 (\alpha_6+3 \rho )\big) \sigma ^2+\alpha_5^3 \mu_1+3 \alpha_5 \mu_1^2 \mu_4\\
 &-\alpha_6 \mu_1^2 \mu_4-\alpha_5^2 \alpha_6 \mu_4+3 \alpha_5^2 \mu_1 \mu_4-2 \alpha_5 \alpha_6 \mu_1 \mu_4+\big(\alpha_5^3+4 \mu_1 \alpha_5^2+3 \mu_4 \alpha_5^2+3 \mu_1^2 \alpha_5\\
 &-2 \alpha_6 \mu_4 \alpha_5+6 \mu_1 \mu_4 \alpha_5-2 \alpha_6 \mu_1 \mu_4+\alpha_6 \mu_4 \rho -\mu_1 (\alpha_5+\mu_1+3 \mu_4) \rho \big) \sigma +\alpha_3 \alpha_4 \big((\mu_1+\sigma )^2+\alpha_6 \mu_4\\
 &+\alpha_5 (2 \mu_1+\mu_4+2 \sigma )\big)\big) \alpha_1^2+\big((\alpha_3 \alpha_4+(\alpha_5+\mu_1+\sigma ) (\alpha_5+\mu_1+2 \mu_4+\sigma )) \alpha_2^3\\
 &+\big(\alpha_5^3+4 (\mu_1+\mu_4+\sigma ) \alpha_5^2+\big(4 \mu_1^2+8 (\mu_4+\sigma ) \mu_1+3 \mu_4^2+4 \sigma ^2+8 \mu_4 \sigma -2 \rho  \sigma \big) \alpha_5\\
 &+\mu_1^3+\sigma ^3+\alpha_6 \mu_4^2+3 \mu_1 \mu_4^2+(3 \mu_1+4 \mu_4-2 \rho ) \sigma ^2+4 \mu_1^2 \mu_4+\big(3 \mu_1^2+8 \mu_4 \mu_1+3 \mu_4^2-\\
 &2 (\mu_1+\mu_4) \rho \big) \sigma +\alpha_3 \alpha_4 (2 \alpha_5+\mu_1+\mu_4+\sigma )\big) \alpha_2^2+\big(2 (\mu_1+\mu_4+\sigma ) \alpha_5^3+\big(4 \mu_1^2+8 (\mu_4+\sigma ) \mu_1\\
 &+3 \mu_4^2+4 \sigma ^2-\alpha_6 \mu_4+8 \mu_4 \sigma -\rho  \sigma \big) \alpha_5^2+\big(2 \mu_1^3+(8 \mu_4+6 \sigma ) \mu_1^2+\big(6 \mu_4^2-2 \alpha_6 \mu_4+16 \sigma  \mu_4\\&+6 \sigma ^2-4 \rho  \sigma \big) \mu_1-\alpha_6 \mu_4 (\mu_4+2 \sigma )+2 \sigma  \big(3 \mu_4^2-\rho  \mu_4+4 \sigma  \mu_4+\sigma ^2-2 \rho  \sigma \big)\big) \alpha_5+2 (\mu_4-\rho ) \sigma ^3\\
 &+\alpha_3^2 \alpha_4^2+\big(\rho ^2-4 (\mu_1+\mu_4) \rho -\alpha_6 \mu_4+3 \mu_4 (2 \mu_1+\mu_4)\big) \sigma ^2+\mu_4 \big(\mu_1^2 (2 \mu_1+3 \mu_4)\\
 &-\alpha_6 \big(\mu_1^2+\mu_4 \mu_1-\mu_4^2\big)\big)-\big(-6 \mu_1 \mu_4 (\mu_1+\mu_4)+\big(2 \mu_1^2+4 \mu_4 \mu_1+3 \mu_4^2\big) \rho\\
 & +\alpha_6 \mu_4 (2 \mu_1+\mu_4+\rho )\big) \sigma +\alpha_3 \alpha_4 \big(2 \alpha_5^2+4 \mu_1 \alpha_5+3 \mu_4 \alpha_5+\mu_1^2+\sigma ^2-\alpha_6 \mu_4+2 \mu_1 \mu_4\\
 &+2 (2 \alpha_5+\mu_1+\mu_4+\rho ) \sigma \big)\big) \alpha_2+\big(\alpha_5^2+2 \mu_4 \alpha_5+\rho ^2-\alpha_6 \mu_4-2 (\alpha_5+\mu_4) \rho \big) \sigma ^3\\
 &+\big(\alpha_5^3+(3 \mu_1+4 \mu_4-\rho ) \alpha_5^2+\mu_4 (-\alpha_6+6 \mu_1+3 \mu_4) \alpha_5-2 (2 \mu_1+\mu_4) \rho  \alpha_5-\alpha_6 \mu_4 (3 \mu_1+\mu_4+\rho )\\
 &+\rho  \big(-3 \mu_4^2-4 \mu_1 \mu_4+\mu_1 \rho \big)\big) \sigma ^2+(\alpha_5+\mu_1) \big(\big(\mu_1^2+2 \mu_4 \mu_1-\alpha_6 \mu_4\big) \alpha_5^2\\
 &+\mu_4 \big(2 \mu_1^2+3 \mu_4 \mu_1-\alpha_6 \mu_4\big) \alpha_5-\alpha_6 \mu_1 \mu_4 (\mu_1+\mu_4)\big)+\big(2 (\mu_1+\mu_4) \alpha_5^3\\
 &+\big(3 \mu_1^2+8 \mu_4 \mu_1-\rho  \mu_1+3 \mu_4^2-\alpha_6 \mu_4\big) \alpha_5^2-2 (\alpha_6-3 \mu_1) \mu_4 (\mu_1+\mu_4) \alpha_5\\
 &+(\alpha_6 \mu_4-2 \mu_1 (\mu_1+\mu_4)) \rho  \alpha_5-\mu_4 \big(\mu_1 (2 \mu_1+3 \mu_4) \rho +\alpha_6 \big(3 \mu_1^2+2 \mu_4 \mu_1+\rho  \mu_1-\mu_4 \rho \big)\big)\big) \sigma \\
 &+\alpha_3^2 \alpha_4^2 (\mu_1+\mu_4+\sigma )+\alpha_3 \alpha_4 \big(\sigma ^3+(3 \mu_1+\mu_4+2 \rho ) \sigma ^2-\alpha_6 \mu_4 \sigma +\mu_1 (3 \mu_1+2 (\mu_4+\rho )) \sigma \\
 &+(\mu_1+\mu_4) \big(\mu_1^2-\alpha_6 \mu_4\big)+\alpha_5^2 (2 \mu_1+\mu_4+2 \sigma )+\alpha_5 \big(2 \mu_1^2+(3 \mu_4+4 \sigma ) \mu_1+2 \sigma ^2\\
 &+\mu_4 (\alpha_6+2 \mu_4)+3 \mu_4 \sigma \big)\big)\big) \alpha_1+(\alpha_3 \alpha_4 \alpha_5+(\alpha_5+\mu_4) (\alpha_5-\rho ) (\mu_4-\rho )) \sigma ^3\\
 &+\big(\alpha_3 \alpha_4 \alpha_5 (\alpha_5+3 \mu_1+\mu_4+2 \rho )+(\alpha_5+\mu_4) \big(\mu_1 \rho ^2-\big(\alpha_5^2+2 \mu_1 \alpha_5+\mu_4 (2 \mu_1+\mu_4)\big) \rho\\
 & +\alpha_5 \mu_4 (\alpha_5+3 \mu_1+\mu_4)\big)\big) \sigma ^2+\alpha_5 (\alpha_3 \alpha_4+\mu_1 (\alpha_5+\mu_1)) (\mu_1+\mu_4) (\alpha_3 \alpha_4+\mu_4 (\alpha_5+\mu_4))\\
 &+\big(\alpha_5 (\alpha_3 \alpha_4+\mu_4 (\alpha_5+\mu_4)) (\alpha_3 \alpha_4+\alpha_5 (2 \mu_1+\mu_4)+\mu_1 (3 \mu_1+2 \mu_4))-\mu_1 \big((\alpha_5+\mu_4) \big(\alpha_5^2\\
 &+\mu_1 \alpha_5+\mu_4 (\mu_1+\mu_4)\big)-2 \alpha_3 \alpha_4 \alpha_5\big) \rho \big) \sigma +\alpha_2^3 (\alpha_3 \alpha_4 \alpha_5+(\alpha_5+\mu_4) (\alpha_5+\mu_1+\sigma ) (\mu_1+\mu_4+\sigma ))\\
 &+\alpha_2^2 \big(\alpha_3 \alpha_4 \alpha_5 (\alpha_5+\mu_1+\mu_4+\sigma )+(\alpha_5+\mu_4) \big(\sigma ^3+(2 \alpha_5+3 \mu_1+2 \mu_4-2 \rho ) \sigma ^2\\
 &+\big(\alpha_5^2+4 \mu_1 \alpha_5+3 \mu_4 \alpha_5+3 \mu_1^2+\mu_4^2+4 \mu_1 \mu_4-(\alpha_5+2 \mu_1+\mu_4) \rho \big) \sigma \\
 &+(\alpha_5+\mu_1) (\mu_1+\mu_4) (\alpha_5+\mu_1+\mu_4)\big)\big)+\alpha_2 \big(\alpha_3^2 \alpha_5 \alpha_4^2+\alpha_3 \alpha_5 \big(\sigma ^2+2 (\alpha_5+\mu_1+\mu_4+\rho ) \sigma \\
 &+(\mu_1+\mu_4) (2 \alpha_5+\mu_1+\mu_4)\big) \alpha_4+(\alpha_5+\mu_4) \big((\alpha_5+\mu_4-2 \rho ) \sigma^3+\big(\alpha_5^2+3 \mu_1 \alpha_5+3 \mu_4 \alpha_5+\mu_4^2\\
 &+\rho ^2+3 \mu_1 \mu_4-2 (\alpha_5+2 \mu_1+\mu_4) \rho \big) \sigma ^2+\big((2 (\mu_1+\mu_4)-\rho ) \alpha_5^2+\big(3 \mu_1^2+6 \mu_4 \mu_1-2 \rho  \mu_1+2 \mu_4^2\big) \alpha_5\\
 &+\mu_1 \mu_4 (3 \mu_1+2 \mu_4)-\big(2 \mu_1^2+2 \mu_4 \mu_1+\mu_4^2\big) \rho \big) \sigma +(\alpha_5+\mu_1) (\mu_1+\mu_4) (\mu_1 \mu_4\\
 &+\alpha_5 (\mu_1+\mu_4))\big)\big)\Bigg)\Big/\Bigg((\alpha_2+\alpha_5+\mu_1+\sigma ) \alpha_1^2+\big(\alpha_2^2+2 (\alpha_5+\mu_1+\mu_4+\sigma ) \alpha_2+\alpha_5^2+\alpha_3 \alpha_4\\
 &-\alpha_6 \mu_4+2 \alpha_5 (\mu_1+\mu_4+\sigma )+(\mu_1+\sigma ) (\mu_1+2 \mu_4+\sigma )\big) \alpha_1+\alpha_5 \mu_1^2+\alpha_5 \mu_4^2+\mu_1 \mu_4^2\\
 &+(\alpha_5+\mu_4-\rho ) \sigma ^2+\alpha_3 \alpha_4 \alpha_5+\alpha_5^2 \mu_1+\alpha_5^2 \mu_4+\mu_1^2 \mu_4+2 \alpha_5 \mu_1 \mu_4+(\alpha_5+\mu_4) (\alpha_5+2 \mu_1+\mu_4) \sigma\\
 & -\mu_1 \rho  \sigma +\alpha_2^2 (\alpha_5+\mu_1+\mu_4+\sigma )+\alpha_2 \big((\alpha_5+\mu_1+\mu_4)^2+2 \sigma  (\alpha_5+\mu_1+\mu_4)+\sigma ^2-\rho  \sigma \big)\Bigg)
   \end{align*}
\section*{Acknowledgements}
This research was supported by project TKP2021-NVA-09, implemented with the support provided by the Ministry of Innovation and Technology of Hungary from the National Research, Development and Innovation Fund, financed under the TKP2021-NVA funding~scheme. I.N. was supported
by the Stipendium Hungaricum scholarship with Application No.~403679. A.D. was supported by the National Laboratory for Health Security, RRF-2.3.1-21-2022-00006 and by the project No.~129877, implemented with the support provided from the National Research, Development and Innovation Fund of Hungary, financed under the KKP\_19 funding
scheme. A.T. was supported by UAEU UPAR grant number 12S125.
\section*{Data availability statement}
Data sharing is not applicable to this article as no new data were created or analyzed in this study.

\end{document}